\documentclass[a4paper,twocolumn,10pt,noarxiv]{quantumarticle}

\usepackage[numbers,sort&compress]{natbib}

\usepackage[T1]{fontenc}
\usepackage{amsmath,amssymb,amsthm,mathtools}
\usepackage{graphicx}
\usepackage{booktabs}
\usepackage{hyperref}
\usepackage{xurl}
\usepackage{cleveref}
\usepackage{enumitem}
\usepackage[normalem]{ulem}
\usepackage{bm}

\crefname{remark}{Remark}{Remarks}
\Crefname{remark}{Remark}{Remarks}
\crefname{proposition}{Proposition}{Propositions}
\Crefname{proposition}{Proposition}{Propositions}

\newtheorem{lemma}{Lemma}
\newtheorem{proposition}{Proposition}
\newtheorem{remark}{Remark}

\newcommand{\Z}{\mathbb{Z}}

\newcommand{\Nmoves}{N_{\mathrm{moves}}}
\newcommand{\Dtw}{D_{\mathrm{tweezer}}}
\newcommand{\Tjerk}{T_{\mathrm{jerk}}}

\title{Efficient atom rearrangements for quantum error correction primitives with a single AOD}

\author{Tom Hartweg}
\affiliation{QPerfect SAS, European Center for Quantum Sciences, 23 rue du Loess, Strasbourg, 67200, France}
\affiliation{University of Strasbourg and CNRS, CESQ and ISIS (UMR 7006)}
\email{tom.hartweg@qperfect.io}
\thanks{Corresponding author.}
\orcid{0009-0006-6600-230X}
\author{Asier Pi\~neiro Orioli}
\orcid{0009-0001-0735-4421}
\affiliation{QPerfect SAS, European Center for Quantum Sciences, 23 rue du Loess, Strasbourg, 67200, France}
\author{Hugo Perrin}
\orcid{0000-0002-1313-8549}
\affiliation{QPerfect SAS, European Center for Quantum Sciences, 23 rue du Loess, Strasbourg, 67200, France}
\author{Samuel Crew}
\affiliation{QPerfect SAS, European Center for Quantum Sciences, 23 rue du Loess, Strasbourg, 67200, France}

\begin{document}

\begin{abstract}
Neutral-atom quantum computers offer arbitrary connectivity enabled by atom transport. Some logical operations can then be simplified or reduced entirely to geometric rearrangements of the atoms. Minimizing the duration of these movements is therefore essential for high logical throughput. We introduce new primitives to shear, rotate and reflect 2D arrays of atoms in a static lattice using sweeps of a single dynamic crossed acousto-optic deflector (AOD) pair. Using (nega-)binary and geometric decompositions, we achieve an AOD stroke count scaling logarithmically in the linear size of the array. In one example, we use the Paeth decomposition to implement a $90^{\circ}$ rotation for a transversal Hadamard gate in a rotated surface code of distance $d$ in $3\lfloor\log_2(d-1)\rfloor + 4$ AOD strokes and $O(d^{1/3})$ constant-jerk time, against $O(d^2)$ strokes and $O(d^{7/3})$ time for atom-by-atom rearrangement.
\end{abstract}

\begin{table*}[t]
    \centering
    
    \small
    \begin{tabular}{lccc}
    \toprule
    Scheme & $\Nmoves$ & $\Dtw$ & $\Tjerk$ \\
    \midrule
    \textbf{$90^\circ$ Rotation (transversal \textsf{H}):} & & & \\
    \quad Sequential & $2d(d-1) + 2$ & $\tfrac{2}{3} (d^2-1) d$ & $O(d^{7/3})$ \\
    \quad Paeth + block & $4(d-1)$ & $4(d-1)$ & $4\alpha(d-1)$ \\
    \quad Paeth + binary & $\bm{3 \lfloor\log_2(d-1)\rfloor + 4}$ & $\lessapprox 8d - 13$
                         & $\lessapprox \bm{16\alpha(d-1)^{1/3}}$ \\
    \midrule
    \textbf{Fold-transversal \textsf{S}:} & & & \\
    \quad Sequential & $4 (d^2 - 2d + 1)$ & $\tfrac{2 d(d-1)(2d-1)}{3}$ & $O(d^{7/3})$ \\
    \quad Paeth + block (\emph{compact})& $10(d-1) - 1$ & $12(d-1)$ & $O(d)$ \\
    \quad Paeth + binary (\emph{compact})& $\lessapprox d + 4 \lfloor\log_2(d-1)\rfloor + 8$
                         & $\lessapprox 12(d-1)$ & $O(d)$ \\
    \quad Paeth + binary (\emph{logarithmic})& $\lessapprox \bm{6 \lfloor\log_2(d-1)\rfloor + 12}$
    & $\lessapprox 16 d - 18$ & $\bm{ O(d^{1/3})}$ \\
    \bottomrule
    \end{tabular}
    \caption{Resource cost of realising the transversal Hadamard and fold-transversal \textsf{S} gate on the rotated surface code, via the atom movements they require (not counting QEC half-cycles), as a function of the code distance $d$ (odd). The constant-jerk time $\Tjerk$ assumes a unit-length stroke takes time $\alpha$. For simplicity, each entry preceded by a $\lessapprox$ states that the exact value scales as the entry and is higher-bounded by it. }
    \label{tab:surfacecost}
\end{table*}

\section{Introduction}

Neutral atom quantum processors are a leading platform for fault-tolerant quantum computation that has recently demonstrated high-fidelity operations~\cite{evered2023, radnaev2025universal, tsai2025benchmarking, peper2025spectroscopy, Senoo2026highfidelity, tao2026finestructure}, high qubit numbers~\cite{Pause2024supercharged, Manetsch2025tweezer6100, Chiu2025continuous} and basic primitives of quantum error correction (QEC)~\cite{Bluvstein2024, SalesRodriguez2025, reichardt2025ftqc, computing2026toric, bluvstein2026ftqcarchitecture, Zhang2026erasure, mathiot2026differentialequations, rines2026demonstrationlogicalarchitecture}.
A defining feature of the platform is the ability to dynamically reconfigure the atoms' positions mid-circuit~\cite{bluvstein2022quantum, rines2026demonstrationlogicalarchitecture, mathiot2026differentialequations, barredo2016atombyatom}, enabling arbitrary connectivity between qubits.
Beyond relaxing connectivity constraints, this reconfigurability can substantially reduce the overhead of logical operations: several logical gates can be simplified, or reduced entirely to geometric rearrangements of the atoms, rather than compiled into sequences of physical entangling gates.
The cost then shifts onto the movement itself, which can be expensive in both time and fidelity; exploiting this at scale therefore requires efficient movement schedules.

Several QEC primitives exemplify this. Beyond the transversal \textsf{CNOT} already demonstrated experimentally~\cite{Bluvstein2024}, the transversal Hadamard on the rotated surface code requires a $90^\circ$ rotation of the data block~\cite{horsman2012surface, chen2026transversal}; the fold-transversal \textsf{S} requires gates between mirror pairs $(i,j) \leftrightarrow (j,i)$, i.e. links running along the $45^\circ$ diagonal. In higher-rate quantum low-density parity-check (LDPC) codes, some \emph{code automorphisms}~\cite{sayginel2025fault}---a stabilizer-preserving permutation of the data qubits---that realise gates on logically encoded qubits likewise require related physical transformations of qubit locations.

There exist different technologies to move atoms which differ in terms of speed, fidelity and flexibility. The current leading method is based on acousto-optical deflectors (AODs), which generate optical tweezers by diffracting a laser beam passing through a crystal using radio-frequency acoustic tones and movement by sweeping these tone frequencies.
This allows to move atoms at $\sim\mu\text{m}/\mu\text{s}$ velocities~\cite{bluvstein2022quantum, Manetsch2025tweezer6100}, and to produce multiple well-calibrated tweezer spots to move atoms in parallel, although subject to geometric constraints. In contrast, spatial light modulators (SLMs) enable high flexibility since they can produce arbitrary trap patterns and move atoms independently in parallel, but their speeds are currently limited to tens of milliseconds per move and fidelities can be sub-optimal due to trap inhomogeneities~\cite{nogrette2014holographic, Kim2016dynamicslm, Knottnerus2025assembly}. Optical lattices are a mature technology enabling high fidelity trapping but movement is generally restricted to collective shifts~\cite{mandel2003opticallattice, Chiu2025continuous}.

In this work, we develop efficient 2D AOD movement schedules for the QEC primitives above, reducing the number of independent moves and the time required to execute the corresponding logical operations. We take AOD constraints into account, in particular the restriction to addressing only rectangular grids of atoms at once. We further assume every move to be a rigid translation of all the atoms addressed along only the $x$ or the $y$ axis. Despite these constraints, we implement a transversal Hadamard on a rotated surface code of distance $d$ in $3\lfloor\log_2(d-1)\rfloor + 4$ strokes, and a fold-transversal S gate in $O(log(d))$ and less than $6 \lfloor\log_2(d-1)\rfloor + 12$ strokes, and $O(d^{1/3})$ constant-jerk time for both (\Cref{sec:rotation}). Those translate directly to the Bacon-Shor gadget for the addressable gate-based computation for La-Cross codes~\cite{pecorari2025addressable}. We also address some code automorphisms in the iceberg and toric codes. Importantly, our solutions require only a \emph{single} crossed AOD pair and no non-rigid movements, which can lead to interferences~\cite{bluvstein2026ftqcarchitecture}. We list the cost of these movements in the atom-by-atom scheme against our proposed schemes in \Cref{tab:surfacecost}.

Previous work includes Xu \emph{et al.}~\cite{xu2024constant}, who reduce 1D rearrangement to $O(\log \ell)$ moves using non-rigid movements, generalised to 2D by~\cite{constantinidesOptimalRoutingProtocols2024} with $O(d\log d)$ asymptotic cost.
Chen \emph{et al.}~\cite{chen2026transversal} propose to realise the rotated-surface-code Hadamard and fold-transversal \textsf{S} using a second AOD pair oriented at $45^\circ$; Sunami \emph{et al.}~\cite{sunami2025transversal} propose instead SLM rearrangement or rotating optical lattices.

\section{Method}
\label{sec:method}

In \Cref{sec:hardware} we outline the hardware assumptions underlying our AOD movement model, which we present in \Cref{sec:model} along with a lower bound on the number of strokes required to realise any prescribed set of displacements. \Cref{sec:shear} introduces the shear of a rectangular sub-array as a primitive and gives a logarithmic stroke implementation. In \Cref{sec:rotation} we compose three such shears via Paeth's decomposition to realise the $90^\circ$ rotation of an array. In \Cref{sec:negation}, we discuss binary decompositions of reflections and in \Cref{sec:45rotation} we discuss 45$^\circ$ reflections and rotations. 

\subsection{Hardware assumptions}
\label{sec:hardware}

We assume a hardware model with two types of atom traps: a \emph{static} 2D square lattice with lattice spacing $\lambda$ (generated by, e.g., an SLM, an AOD pair, or an optical lattice) and a \emph{dynamic} array generated by a 2D crossed AOD pair.
The dynamic array can reconfigure the atomic positions by picking up, moving and releasing sets of atoms. Each such \emph{stroke} must respect AOD constraints and other assumptions we make:
\begin{enumerate}
    \item The dynamic AOD pair is oriented along the same $x-y$ axes of the static lattice and can address any rectangular sub-array resulting from the intersection of a set of rows and columns.
    \item Atoms are only moved along $x$ or $y$ at once (no diagonal moves will be necessary for our protocols).
    \item All atoms in the selected sub-array are moved in block by the same amount. In principle, AODs allow for squeezing and stretching moves and independent movements of rows/columns of the trapped atoms (as long as rows/columns do not cross~\cite{xu2024constant}), but this is not necessary in our proposal, which is useful as it can lead to interference due to intermodulations~\cite{bluvstein2026ftqcarchitecture}.
    \item All atoms present in the lattice sites aligned with the dynamic AOD pair sites during pickup will be transferred. The inverse applies for release. The exact picking or release mechanism depends on the static array's technology, but typically involves aligning the dynamic and static arrays and increasing or decreasing the AOD's trap depth.
\end{enumerate}
To further ensure that atom coherence is preserved during the transport process collisions must be avoided. We assume:
\begin{enumerate}[resume]
    \item Collisions with both occupied and empty static traps during movement can always be avoided. When moving a full row (column) along $x$ ($y$), collisions are naturally avoided if that row's (column's) static traps can be turned off. If that is not possible, we can assume that movement happens at $\lambda/2$ in between lattice sites, which should be possible for traps generated by different lasers or if $\lambda$ is sufficiently large.
    \item Collisions at release time are avoided by ensuring that no atom occupies a static site overlapping with a dynamic trap, even if the latter is empty.
\end{enumerate}
We will be interested in quantifying the number of strokes, total distance and time necessary to implement several 2D transformations, for which we assume that:
\begin{enumerate}[resume]
    \item We may ignore the time necessary to transfer atoms between the static and dynamic arrays, as we expect time to be dominated by movement. Transfer time contribution scales linearly with the number of pickups and releases.
    \item Each move follows a constant-jerk trajectory to minimize motional excitation (the generalization to other movement schemes is easy).
    \item The potential extra $\lambda/2$ strokes necessary for avoiding collisions will be ignored for timing and distance, but the additional cost scales linearly with the number of strokes.
\end{enumerate}

\subsection{AOD movement model and lower bound}
\label{sec:model}

With the above assumptions, we model the static traps as a uniform integer lattice $\Z^2$ with at most one atom per site $(x,y)\in \Z^2$, and identify an \emph{atom array} $A$ with an ordered tuple of occupied sites, such that $A_i$ is the coordinate of atom $i$. We denote by $|A|$ the number of elements in $A$, i.e.~the number of atoms. A single 2D AOD stroke is a triple $M = (S_x, S_y, m)$, where $S_x, S_y \subset \Z$ are column and row selections (the AOD addresses every site in $S_x \times S_y$) and $m \in \Z \times \{0\} \cup \{0\} \times \Z$ is a horizontal or vertical displacement; the stroke sends each $a \in A \cap (S_x \times S_y)$ to $a + m$, and is \emph{valid} if the resulting array has at most one atom per site.
A \emph{movement scheme} is a sequence $(M_1, \dots, M_n)$ of strokes, with respective \emph{stroke displacements} $(m_1, \dots, m_n)$. We define the metrics $\Nmoves = n$, $\Dtw = \sum_i |m_i|$, and --- assuming the constant-jerk trajectories typical of high-fidelity AOD control~\cite{bluvstein2022quantum}, in which a stroke of length $\ell$ takes time $\alpha \ell^{1/3}$ with $\alpha = (12/j\lambda)^{1/3}$ ---
$\Tjerk = \alpha \sum_i |m_i|^{1/3}$.

\paragraph{Lower bound.}
We now prove a lower bound on $\Nmoves$ and $\Dtw$. For a starting array $A_s$ and a target $A_f$ of the same size ($|A_s| = |A_f|$), let $\Delta = \{A_{f,i} - A_{s,i} | i\in 1,...,|A_s|\} \setminus \{0\}$ be the set of required displacements for each atom, $\Delta_x, \Delta_y$ the $x$- and $y$-components of the required displacements, and let $\max_{\pm,x/y} = \max(\pm\Delta_{x/y} \cup \{0\})$ denote the maximal displacement among all atoms in every direction, $\pm x$ and $\pm y$. We then have the following:

\begin{lemma}\label{lem:lowerbound}
For any valid movement scheme realising $A_s \to A_f$,
\begin{align}
\Nmoves &\ge \lceil \log_2(|\Delta_x|+1) \rceil
         + \lceil \log_2(|\Delta_y|+1) \rceil,
         \label{eq:LBnmoves}\\
\Dtw    &\ge \max_{+,x} + \max_{-,x} + \max_{+,y} + \max_{-,y},
         \label{eq:LBdtw}
\end{align}
\end{lemma}
\noindent
where $|\Delta_{x/y}|$ is the number of non-zero required displacements along $x/y$. As it will be used throughout the paper, we point out that if $|\Delta_{x/y}|>0$, we have $\lceil \log_2(|\Delta_{x/y}|+1) \rceil = \lfloor \log_2(|\Delta_{x/y}|)\rfloor +1 $.

\begin{proof}[Proof]
We obtain a lower bound by relaxing the constraints to permit multiple atoms per site and arbitrary atom selections (rather than only rectangular ones). In this setting horizontal and vertical strokes commute, and reordering them changes neither $\Nmoves$ nor $\Dtw$, so any scheme can be brought to the form of $n_x$ horizontal strokes followed by $n_y$ vertical strokes.

Consider the horizontal strokes, with stroke displacement list $(m_1, \dots, m_{n_x})$. Each atom is either selected or not in each stroke, so its net horizontal displacement is a subset sum $\sum_{i} b_i m_i$ with $b_i \in \{0,1\}$, taking at most $2^{n_x}$ distinct values. Since the scheme realises $A_s \to A_f$, every displacement in $\Delta_x$ arises this way as a subset sum; together with $0$ (the empty selection, excluded from $\Delta_x$) that is $|\Delta_x| + 1$ distinct values among at most $2^{n_x}$. Hence $2^{n_x} \ge |\Delta_x| + 1$, i.e. $n_x \ge \lceil \log_2(|\Delta_x|+1) \rceil$; the same argument applied to the vertical strokes gives $n_y \ge \lceil \log_2(|\Delta_y|+1) \rceil$, which gives
\cref{eq:LBnmoves}.

For \cref{eq:LBdtw}, every atom displaced by distance $d$ in one of the four cardinal directions forces the tweezer to travel at least $d$ in that direction, so $\Dtw$ is bounded from below by the sum of the maximal required displacements in each of the four directions.
\end{proof}

\Cref{lem:lowerbound} can be tight in the relaxed constraints setting. For $\Delta_x = \Delta_y = \{1, \dots, n\}$, the binary stroke displacement list $(2^0, 2^1, \dots, 2^{\lfloor\log_2 n\rfloor})$ generates every required displacement as a subset sum with exactly $\lfloor\log_2 n\rfloor + 1$ strokes per axis, saturating~\cref{eq:LBnmoves}.
For $\Delta_x = \Delta_y = \{-a, \dots, b\} \setminus \{0\}$, the negabinary list $(2^0, -2^1, 2^2, \dots)$ does the same.
The distance bound of~\cref{eq:LBdtw} is saturated for $\Delta_x = \{d_1, \dots, d_n\}$ (analogously for $\Delta_y$) with $d_k<d_{k+1}$ and $d_1>0$ by a stroke displacement list $(d_1,d_2-d_1, d_3-d_2,\ldots)$ such that the subset sum $\sum_i b^{(k)}_i m_i$ of atom $k$ has $b^{(k)}_i=1$ iff $i\lessapprox k$; in the presence of positive and negative displacements $d_k$ one gets two analogous sets of stroke displacement lists.
Whether these strokes can be implemented with the full AOD constraints and without collisions is strongly dependent on the movement scheme.

\subsection{The shear primitive}
\label{sec:shear}

\paragraph{Definition.}
A \emph{horizontal shear} of an $\ell \times d$ array is the affine map $(i, j) \mapsto (i + m(j - h),\, j)$ with reference row $h \in \Z$ and rate $m \in \Z$; vertical shears are defined analogously. Although each row of the static lattice is a valid AOD rectangle, distinct rows must be translated by the distinct amounts $m, 2m, \dots, (d-1)m$ (the maximum depends on $h$), so a shear requires multiple strokes. A change of reference line can be factored out as a global shift of the $x$ coordinates.

\paragraph{Three schemes.}
We introduce three example implementations using horizontal AOD strokes. We assume that the reference line is the top row of the array.

\begin{itemize}
\item \textit{Line-by-line.} Translate each row in turn. \emph{Cost:} one move per row so $\Nmoves = d - 1$. The distance travelled by the tweezer in each move is $1,2,\dots,d-1$, so $\Dtw = |m| d(d-1)/2$.

\item \textit{Block.} Stroke displacement list $(m, m, \dots, m)$ of length $d-1$; at stroke~$k$, the AOD selects and moves all rows with target displacement exceeding $km$. \emph{Cost:} the number of moves is the same: $\Nmoves = d-1$. But as this scheme maximises the parallelisation of movements, we saturate the bound of \cref{eq:LBdtw} to get $\Dtw = |m|(d-1)$. 

\item \textit{Binary.} The stroke displacement list is the binary coefficients $(m, 2m, 4m, \dots, 2^{\lfloor\log_2(d-1)\rfloor}m)$; at stroke~$k$, the AOD selects all rows whose target displacement has a~$1$ in bit~$k$ of its binary representation. \emph{Cost:} This methods provides a logarithmically scaling number of moves: $\Nmoves = \lfloor\log_2(d-1)\rfloor + 1$, which saturates the bound of~\cref{eq:LBnmoves}. We also have: $\Dtw = |m|(2^{\Nmoves}-1)$. A convenient higher bound with the same scaling is $|m|(2d-3)$, reached when $d = 2^c +1$ for some $c\in\mathbb{N}$. Throughout the text, we will use $f\lessapprox g$ to state that $f=O(g)$ and $f\le g$. Thus $\Dtw\lessapprox |m|(2d-3)$.
\end{itemize}

\Cref{fig:shears} illustrates the three schemes on a $5 \times 5$ array. For other reference lines, the same three schemes can be used, with the eventual application of a global shift. For interior reference lines, the binary scheme may be replaced by a negabinary variant with the same asymptotic, and the exact same cost for the central row in odd $d$ arrays. All three schemes are collision-free: every stroke is a horizontal translation, so atoms never leave their row, and within a row the whole line is shifted rigidly by a common amount, preserving spacing.

\paragraph{Constant-jerk timing.}
At constant speed, the time taken by each implementation would be proportional to $\Dtw$, thus the block scheme would be faster. However, as described previously, AOD movements are usually done using non-constant velocity profiles such as constant jerk, which gives a displacement time of\footnote{Recall $\alpha = (12/j\lambda)^{1/3}$ for jerk $j$.} $\alpha \ell^{1/3}$ over a distance $\ell$ . In this setting, longer moves are less penalised, and $\Nmoves$ gains in importance. The block scheme uses $d-1$ unit strokes, so it's execution time is $T_{\mathrm{block}} = \alpha(d-1)$. The binary scheme uses $\lfloor\log_2(d-1)\rfloor + 1$ strokes of lengths $1, 2, 4, \dots$, giving an execution time of:
\begin{align}
T_{\mathrm{binary}}
&= \alpha \!\!\sum_{k=0}^{\lfloor\log_2(d-1)\rfloor}\!\!\! 2^{k/3}
 = \alpha \cdot \frac{2^{(\lfloor\log_2(d-1)\rfloor+1)/3} - 1}{2^{1/3} - 1} \notag\\
&\lessapprox4.85\alpha (d-1)^{1/3}.
\label{eq:binarytime}
\end{align}
Clearly $T_{\mathrm{block}} > T_{\mathrm{binary}}$ for all $d \ge 4$; the binary scheme is therefore preferred under constant-jerk for array size above $d=3$.

\begin{figure*}[t]
\centering
\includegraphics[width=\textwidth]{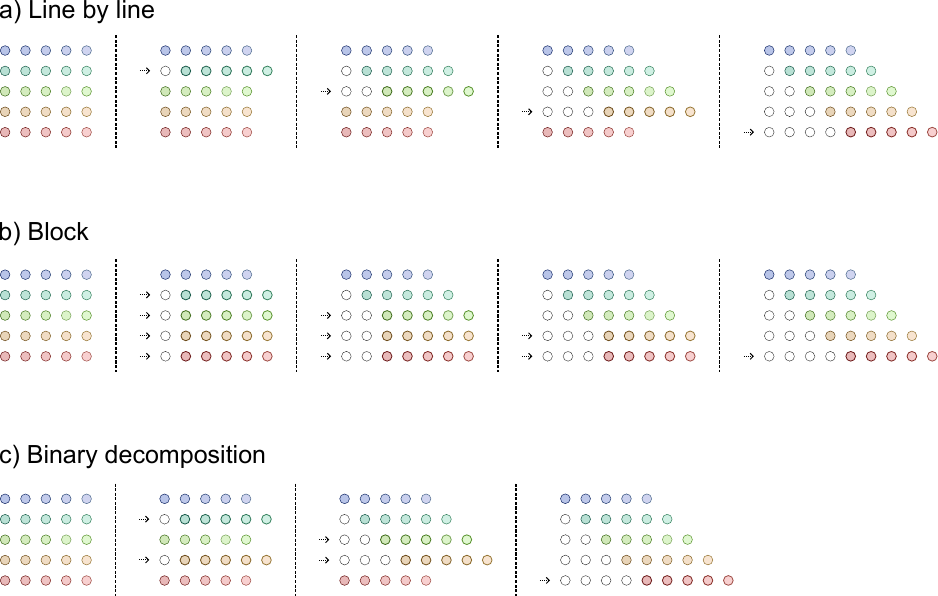}
\caption{The three shear schemes applied to a $5 \times 5$ atom array with the top row as reference. Each panel shows the array after a single AOD  stroke. (a) \emph{Line-by-line} moves one row per stroke ($d - 1 = 4$ strokes). (b) \emph{Block} moves all rows whose target displacement still exceeds the current cumulative shift, four strokes of displacement $1$ each. (c) \emph{Binary decomposition} encodes each row's target displacement in binary: at stroke $k$ the AOD selects rows whose target has a $1$ in bit $k$ and moves them by $2^k$ ($\lfloor\log_2(d-1)\rfloor + 1 = 3$ strokes).}
\label{fig:shears}
\end{figure*}

\subsection{Rotation by shearing}
\label{sec:rotation}

The $90^\circ$ rotation $(x, y) \mapsto (-y, x)$ of a $d \times d$ square array ($d$ odd, centred at the origin) is a standard geometric operation underlying a range of protocols for different QEC schemes, as we detail in \Cref{sec:applications}. If the rotation is performed atom-by-atom, it costs $\Nmoves = O(d^2)$ and $\Dtw = O(d^3)$.

\paragraph{Paeth decomposition.}
Paeth~\cite{paeth1990fast} observed that any rotation factors into three
shears:
\begin{multline}
R(\theta) =
\begin{pmatrix} 1 & -\tan(\theta/2) \\ 0 & 1 \end{pmatrix}
\begin{pmatrix} 1 & 0 \\ \sin\theta & 1 \end{pmatrix} \\
{} \times
\begin{pmatrix} 1 & -\tan(\theta/2) \\ 0 & 1 \end{pmatrix}.
\label{eq:paeth}
\end{multline}
At $\theta = 90^\circ$ this is three shears (horizontal, vertical, horizontal), each an instance of the primitive of
\Cref{sec:shear}. \Cref{fig:rotation} illustrates a $5 \times 5$ array through the three stages. Done naively, the three shears leave a residual translation which may be resolved by either taking the reference line of each shear to be the central row/column (fixing the centre of the array throughout, so no correction is needed), or taking the first and last shears with opposite directions and mirrored reference lines (reducing the correction to a single vertical stroke).

\paragraph{Cost.}
The first and third shears act on $d$ rows of length $d$; the middle shear, on $2d-1$ non-empty columns after the first shear. The binary scheme gives:
\begin{align}
\Nmoves^{\text{rot}} &= 3 \lfloor \log_2(d-1) \rfloor + 4, \notag\\
\Dtw^{\text{rot}} &\lessapprox 8d - 11, \\
\Tjerk^{\text{rot}} &\lessapprox 16\alpha(d-1)^{1/3}. \notag
\end{align}
Applied to the $90^\circ$ rotation, which has
$\Delta_{x/y} = \{-(d-1), \dots, -1, 1, \dots, d-1\}$,
\Cref{lem:lowerbound} reads
$\Nmoves \ge 2\lfloor\log_2(d-1)\rfloor + 4$ and $\Dtw \ge 4(d-1)$. The construction is within a factor $3/2$ of optimal in $\Nmoves$ and a factor $2$ in $\Dtw$. In contrast, the block scheme saturates $\Dtw$ at the cost of returning to linear $\Nmoves$ and $\Tjerk$. 

\begin{figure*}[t]
\centering
\includegraphics[width=1.0\textwidth]{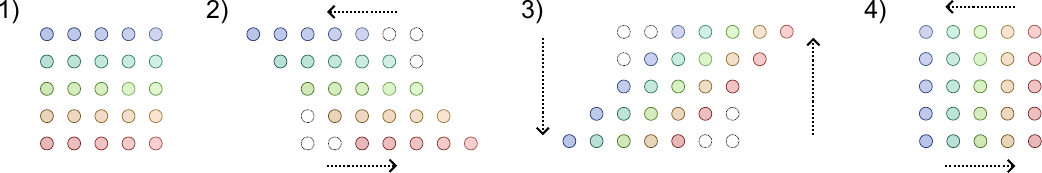}
\caption{Paeth decomposition of the $90^\circ$ rotation of a $5 \times 5$ array into three shears. Panels 1--4 show the array after each of the three shears. Here, for each shear, the reference line is chosen to be the middle row/column.
  }
\label{fig:rotation}
\end{figure*}

\begin{remark}
If the transfer of the atoms between the dynamic to the static array is a bottleneck for the fidelity (as it can be the case for an SLM static lattice with dynamic AOD \cite{Bluvstein2024}), the block scheme provides a method where each atom is at maximum transferred 6 times, independent of d: If a partial transfer of the atoms from the dynamic AOD to the static lattice can be performed, the atoms can be dropped line by line during the block movement in the shears.\footnote{To drop a single line one may need to apply a $\lambda/2$ offset to only one of the rows/columns to align with the static tweezer, see \cref{sec:hardware}.}
\end{remark}

\begin{remark}\label{rmk:subdivision}
In some cases, it is possible to improve on the Paeth + binary protocol by one move by dividing the array into two parts and applying shears in parallel. This amounts, for odd $d$, to a move count of: 
$$
\lfloor \log_2(d-1)\rfloor+2\lfloor \log_2(3d-3)\rfloor+1
$$
and gives a one move improvement (compared with binary + Paeth) whenever the decimal part of $log_2(d-1)$ is strictly smaller than $ 2-\log_2(3)\approx 0.415$. For even $d>2$, the move count is:
$$
\lfloor \log_2(d-2)\rfloor+\lfloor \log_2(3d-4)\rfloor+\lfloor \log_2(3d-2)\rfloor+1
$$
In the $d=3$ case, this protocol gives a six-move solution saturating the theoretical lower bound. The moves are illustrated in figure \ref{fig:d3_opti}. The general protocol is illustrated for $d=9$ in figure \ref{fig:d5_opti}.
\end{remark}

\begin{figure}[t]
    \centering
    \includegraphics[width=0.95\linewidth]{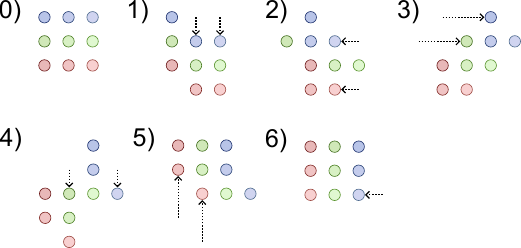}
    \caption{Optimal 6-move 90$^\circ$ rotation for $d=3$.}
    \label{fig:d3_opti}
\end{figure}

\begin{figure*}[t]
    \centering
    \includegraphics[width=0.95\textwidth]{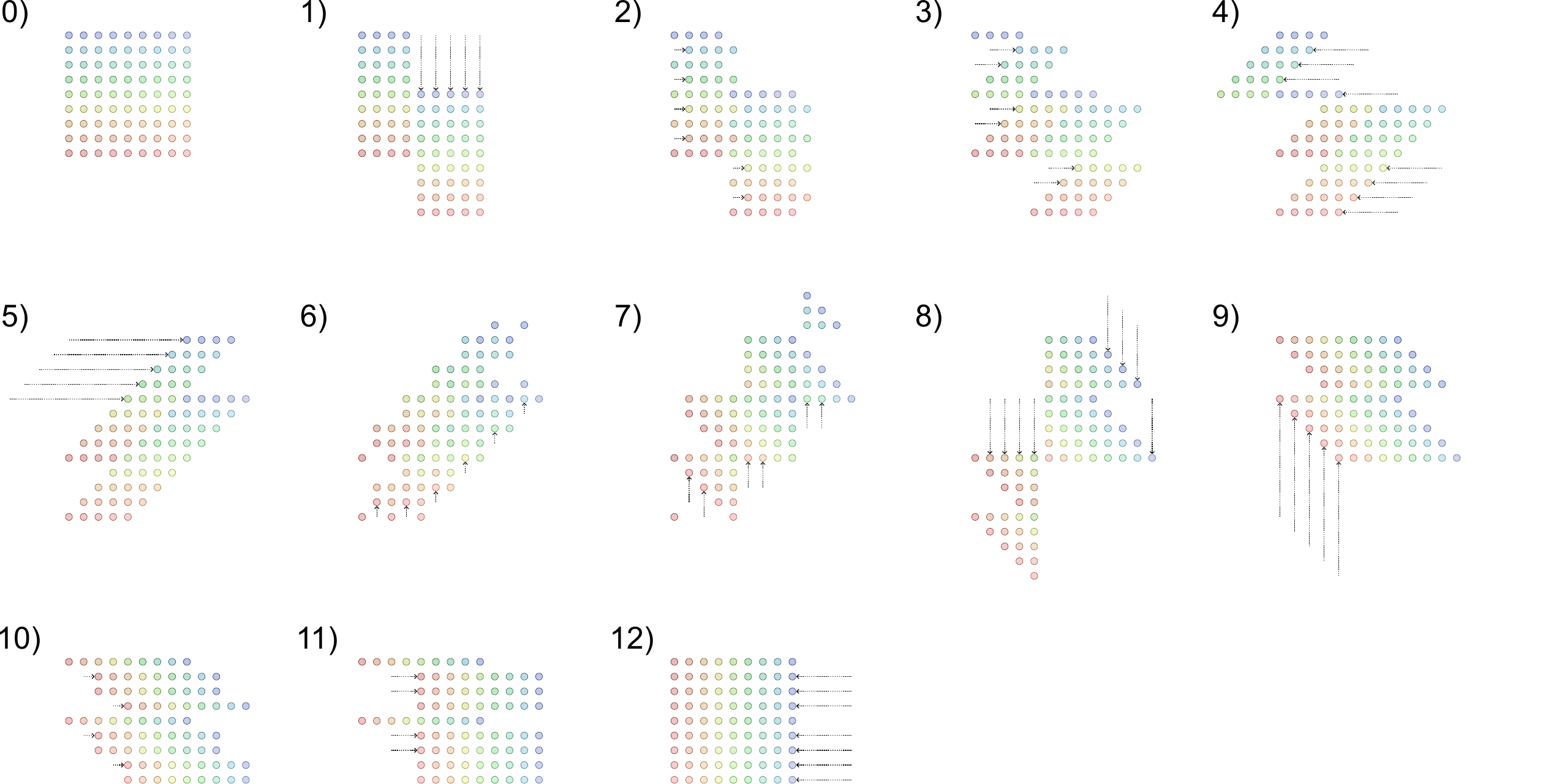}
    \caption{Sub block moves for $d=9$. Twelve moves beats the Paeth+binary protocol by one move.}
    \label{fig:d5_opti}
\end{figure*}

\subsection{Reflection}
\label{sec:negation}
A reflection about the origin is a single-axis rearrangement $x \to -x$ where the atom at coordinate $i$ acquires a displacement $-2i$. The required displacements carry both signs, so we replace the binary scheme with its \emph{negabinary} variant. We write $-2i=\sum_{k\ge0}b_k(i)\,(-2)^k$ with $b_k(i)\in\{0,1\}$; stroke $k$ then translates by $(-2)^k$ every atom with $b_k(i)=1$. The single fixed stroke list $\bigl((-2)^0,(-2)^1,(-2)^2,\dots\bigr)$ realises displacements of either sign. Applied to a line of length $l$, a reflection costs:
\begin{equation}
  \Nmoves=\lceil\log_2(l)\rceil,
  \label{eq:neg-cost}
\end{equation}
saturating the lower bound \eqref{eq:LBnmoves}. The moves are illustrated in \Cref{fig:negabinary_reflection} for $l=5$. Note that avoiding collisions during movement would require in this case small offsets (e.g.~by $\lambda/2$) before and after every move.
\begin{figure}[t]
    \centering
    \includegraphics[width=0.85\linewidth]{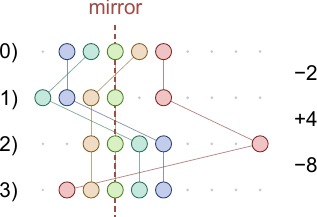}
    \caption{Reflection of a line of 5 atoms across a mirror axis, implemented via negabinary decomposition. The labels on the right indicate the distance of the stroke at each step.}
    \label{fig:negabinary_reflection}
\end{figure}
\begin{proposition}\label{prop:neg-collision}
The axis negation is collision-free after every stroke.
\end{proposition}
\begin{proof}
Atom positions depend only on the set $M$ of strokes already performed; fix $M$, and suppose two atoms collide, $p_M(i)=p_M(i')$ with $i\neq i'$, where $p_M(i)=i+\sum_{k\in M}b_k(i)(-2)^k$ is the position of the atom that started in $i$ after the strokes in $M$ are applied.
Subtracting the two positions,
\begin{equation*}
  i-i'=\sum_{k\in M}c_k(-2)^k,
\end{equation*}
where $c_k:=b_k(i')-b_k(i)\in\{-1,0,1\}$. Subtracting the negabinary expansions of $-2i$ and $-2i'$ instead gives $2(i-i')=\sum_{k\ge0}c_k(-2)^k$. Eliminating $i-i'$ between the two,
\begin{equation*}
  \sum_{k\ge0} c_k(-2)^k\Bigl(\mathbf 1_{k\in M}-\tfrac12\Bigr)=0,
  \quad\text{i.e.}\quad \sum_{k\ge0} d_k\,2^k=0,
\end{equation*}
with $d_k:=c_k(-1)^k\bigl(2\cdot\mathbf 1_{k\in M}-1\bigr)\in\{-1,0,1\}$, where $\mathbf 1_{k\in M}=1$ if $k\in M$ and $0$ otherwise. Were some $d_k$ nonzero, the smallest such index $j$ would satisfy $d_j=-2\sum_{k>j}d_k2^{k-j-1}$, forcing $d_j$ even and contradicting $d_j=\pm1$. Hence $d_k=0$ for all $k$, and since $2\cdot\mathbf 1_{k\in M}-1=\pm1$ never vanishes, $c_k=0$ for all $k$. Thus $-2i=-2i'$, i.e. $i=i'$.
\end{proof}

\begin{remark}
    The $90^\circ$ rotation and the axis-negation can be combined to perform diagonal and anti-diagonal reflections. Indeed, reflection about the $x=y$ diagonal can be decomposed into a vertical axis reflection followed by a $90^\circ$ anti-clockwise rotation. 
\end{remark}

\subsection{\texorpdfstring{$45^\circ$}{45 degrees} rotations}
\label{sec:45rotation}

While rotations are only symmetries of the lattice if they are multiples of 90$^\circ$, one can rotate by other angles if combined with stretches/shears.

A stretched rotation by 45$^\circ$ can be implemented in two steps: (i) a horizontal shear with $|m|=1$ with reference row in the center as for the 90$^\circ$ rotation, and (ii) a vertical shear with $(i,j) \mapsto (i, j+\lfloor (i-v)/2 \rfloor)$ with, e.g., center reference column $v=0$. The latter corresponds to an $|m|=0.5$ shear snapped to the lattice.

Strictly speaking, this step divides the array into even and odd subarrays, $A$ and $B$, and applies a stretched 45$^\circ$ rotation to each one but shifted by 0.5 lattice spacing with respect to one another because of the lattice snapping. After this step, subarray $A$ ($B$) is mapped to only even (odd) columns (or viceversa). The protocol is described in  \Cref{fig:45deg_rotation}.

\begin{figure*}
    \centering
    \includegraphics[width=\linewidth]{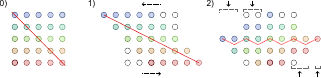}
    \caption{Approximate rotation by 45$^\circ$ decomposed as two shears. The dashed red line represents the diagonal of the original configuration (0). The final configuration (2) shows how the diagonal gets stretched and staggered to snap back to the static lattice.}
    \label{fig:45deg_rotation}
\end{figure*}

\section{Applications}
\label{sec:applications}

We give three illustrative applications of our transformations to quantum error correction (QEC) protocols. In \Cref{sec:surfacecode} we discuss transversal \textsf{H} and fold-transversal \textsf{S} on the surface code (SC), via rotation of the data block. In \cref{sec:autogate} we discuss logical Cliffords from code automorphisms in Iceberg codes \cite{chao2018fault,self2024protecting} and toric codes and in \Cref{sec:bsgadget} addressable Cliffords on La-cross LDPC codes \cite{pecorari2025addressable} via a Bacon-Shor gadget.

\subsection{Surface-code \textsf{H} and fold-transversal \textsf{S}}
\label{sec:surfacecode}

The surface code is arguably the most well-developed and understood QEC code, demonstrating time-efficient universal computing along with fast and performant decoders. This makes the rotated surface code the QEC code of choice for most early fault-tolerant hardware modalities~\cite{fowler2012surface, litinski2019game, ismail2025transversal}. Even in the recently proposed space-efficient architectures \cite{webster2026pinnacle, cain2026shor}, which choose qLDPC codes as their main code, the surface code is still used as an auxiliary gadget for computation, mainly for magic state preparation. We propose here, building on the efficient rotation introduced in \Cref{sec:rotation}, efficient implementations of the logical Hadamard and \textsf{S} gates in a distance $d$ surface code. Two variants of the surface code exist: the regular and the rotated. The regular variant uses $2d^2 - 2d + 1$ data qubits, which can be arranged either in two square arrays of linear size $d$ and $d-1$, a single rectangular array of size $d\times 2d-1$ or in a single square array of size $2d-1$. We adopt the latter arrangement in what follows for simplicity, but note that small gains in stroke counts are possible for the other variants.
The rotated variant requires only $d^2$ data qubits arranged in a  $d\times d$ square array.

\paragraph{Transversal Hadamard.}
A logical Hadamard on the surface code at distance $d$ is a transversal physical Hadamard on the data atoms (free in AOD cost) followed by a $90^\circ$ rotation of the data block to restore the canonical orientation of logical operators~\cite{horsman2012surface,chen2026transversal}. By \Cref{sec:rotation} the rotation costs $3\lfloor\log_2(d-1)\rfloor + 4$ strokes for the rotated variant and $3\lfloor\log_2(d)\rfloor + 7$ for the regular surface code. The constant-jerk time scales as $O(d^{1/3})$  --- a polynomial improvement on the
$O(d^{7/3})$ of a move-by-move implementation.

\paragraph{Fold-transversal \textsf{S} gate.}
The fold-transversal \textsf{S} gate~\cite{moussa2016transversal, breuckmann2024fold} acts on the regular surface code as a single-qubit $\textsf{S}$ or $\textsf{S}^\dagger$ on each atom of a $45^\circ$ diagonal, together with a \textsf{CZ} between every diagonal-mirror pair $(i, j) \leftrightarrow (j, i)$ (\Cref{fig:sgate}a). 
The natural $45^\circ$ \textsf{CZ}-disjoint partition is incompatible with a single
horizontal/vertical AOD and would require a second AOD oriented at $45^\circ$~\cite{chen2026transversal}. A move-by-move implementation costs $4(d^2 - d - 1)$ strokes and $O(d^{7/3})$ time.

We propose to use the $45^\circ$ rotation primitive of \Cref{sec:45rotation} to align the diagonal links with an AOD axis. The protocol has three forward stages followed by their reverse, applied on the $2d-1$ rows/columns of the regular surface code:
\begin{enumerate}
\item Horizontal shear of all $2d - 1$ rows with rate $1$
  (\Cref{fig:sgate}b).
  
\item Vertical shear of the inner $2d - 3$ occupied columns with rate $1$
  (\Cref{fig:sgate}c); the two outermost columns contain no
  \textsf{CZ} endpoints and are skipped. Stages 1 and 2 together
  align the $45^\circ$ diagonal with the horizontal axis, leaving
  each \textsf{CZ}-paired atom directly above its partner.
  
\item We propose two versions of this step: 
\\a) \emph{Logarithmic:} If space permits, the whole upper triangular block can be translated horizontally or vertically away from the other atoms, then reflected using the scheme of \Cref{sec:negation}, and brought back next to the lower triangular block. Then the \textsf{CZ} gates can be performed. This amounts to $\lceil\log_2(d-1) \rceil +3$ moves.
\\ b) \emph{Compact:} Otherwise, translate the upper triangular block downward in parallel, stopping at the prescribed distances to fire each \textsf{CZ} layer (\Cref{fig:sgate}d). This is $d - 1$ strokes of length $2$. Note that an additional unit-length stroke is used to offset the displaced block relative to the static one. 
  
\end{enumerate}
Stages 1, 2 and 3.a) reverse with the same stroke count as their forward counterparts; the inverse of stage 3.b needs only $2$ strokes (no intermediate stops on the return). To implement the fold-transversal $\textsf{S}$ gate on the rotated surface code, one can morph it into the regular surface code by applying half a QEC cycle with ancillas~\cite{chen2026transversal}.

\Cref{tab:surfacecost} collects the costs for the rotated surface code (not counting the QEC half-cycles). The Hadamard and the \textsf{S} gate with scheme 3.a) reach $O(d^{1/3})$ time; the compact $\textsf{S}$ gate (using scheme 3.b)) improves on the sequential $O(d^{7/3})$ by a factor $d^{4/3}$ but saturates at linear scaling, because the diagonal sweep involves $O(d)$ independent mirror-pair links.

\begin{figure*}[t]
\centering
\includegraphics[width=0.7\textwidth]{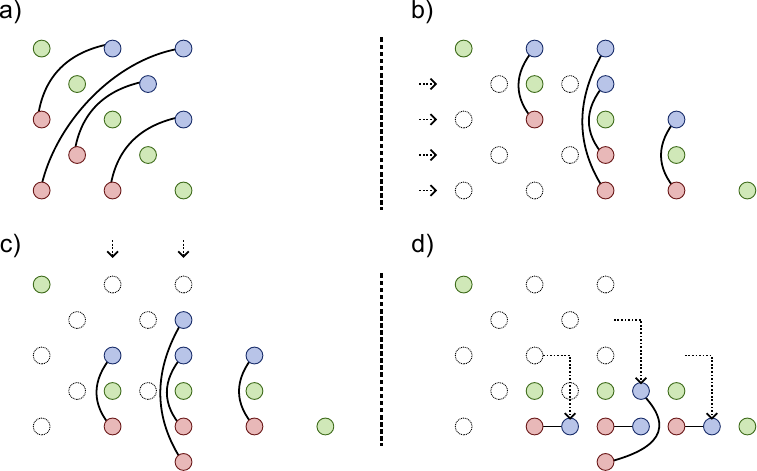}
\caption{The compact fold-transversal \textsf{S}-gate protocol for a distance-$3$ rotated surface code. (a) After half a QEC cycle, the register is in the regular surface code orientation with $2d - 1 = 5$ non-empty rows and columns; blue/red atoms undergo \textsf{CZ} gates (black arcs) across a $45^\circ$ diagonal, green atoms undergo single-qubit $S$ gates. (b) Horizontal shear of all $5$ rows with rate $1$. (c) Vertical shear of the inner $3$ atom-holding columns with rate $1$ (the leftmost and rightmost atoms are unaffected). Stages (b) and (c) together align the diagonal with the vertical axis, and the \textsf{CZ}-paired atoms now sit on common vertical lines. (d) The upper partner of each \textsf{CZ} pair translates downward in parallel, stopping along the way to perform the CZ layers. Reversing stages (a)--(c) restores the original layout. As stated in the main text it is also possible, if space permits, after stage (c) to instead apply the reflection of fig \ref{fig:negabinary_reflection}}.
\label{fig:sgate}
\end{figure*}

\begin{remark}[(Fold-)transversal gates on hypergraph product (HGP) codes] The regular surface code is the simplest instance of an HGP code. The (fold-)transversal Hadamard and $S$ logical operations can be generalised to broader classes of HGP codes encoding multiple logical qubits. In such codes, each logical qubit can be identified with a physical qubit in a top-left sub-array of the code.

For HGP codes where, the transversal Hadamard followed by a $90^{\circ}$ rotation realises a global logical Hadamard and swaps pairs of logical qubits whose identified physical qubits are mirror images across the main diagonal. Similarly, the fold-transversal \textsf{S} gate applies a logical \textsf{S} gate to each logical qubit whose identified physical qubit lies on the main diagonal, and a logical \textsf{CZ} to each mirror pair of logical qubits. HGP codes can be arranged in two square arrays of size $n\times n$ and $r \times r$ with $r \le n$, or in a single array of size $2n-1 \times 2n-1$ if $r < n$ ($2n \times 2n$ if $r = n$). Although less efficient for storage, the single-array arrangement costs only $3\lfloor\log_2(d)\rfloor + 7$ AOD strokes to rotate, whereas rotating each array individually costs $3(\lfloor\log_2(n-1)\rfloor + \lfloor\log_2(r-1)\rfloor) + 8$. The time
overhead and the cost of the fold-transversal gate \textsf{S} are similar to those derived for the rotated surface code, with $d$ replaced by $n$.
\end{remark}

\subsection{Automorphism gates}
\label{sec:autogate}

A code automorphism is a permutation $\sigma$ of the physical qubits that maps the stabiliser group to itself. Because $\sigma$ leaves the stabiliser group invariant, it preserves the codespace and acts on the logical qubits as a (possibly identity) logical Clifford operation. The permutation $\sigma$ may be realised by physically relocating the atoms with AOD movements with no physical gates applied and no associated gate error \cite{grassl2000cyclic,sayginel2025fault}.  We give now a number of examples of shears, negabinary reflections, rotations and Dehn twists in this context.

\paragraph{Example.}
The Iceberg code $[[n,n-2,2]]$ \cite{self2024protecting,chao2018fault} has stabilizers $X^{\otimes n}$ and $Z^{\otimes n}$ invariant under every qubit permutation. Its permutation automorphism group is thus $S_n$. With the qubits arranged in a row, the reflection permutation is a single-axis negation, realised by the negabinary scheme of \Cref{sec:negation}. For even $n$, in the basis $\bar X_a=X_aX_{n-1}$, $\bar Z_a=Z_aZ_0$ ($a=1,\dots,n-2$), it implements the logical Clifford $\bar P_a\mapsto\prod_{b\neq n-1-a}\bar P_b$ ($P\in\{X,Z\}$). For example, with the $[[6,4,2]]$ code this corresponds to $4$ AOD strokes realising a $9$ \textsf{CNOT} circuit.

\paragraph{Example.}
The toric code\footnote{The toric code was recently realised experimentally on a neutral atom platform~\cite{computing2026toric}.} places $\ell^2$ data atoms on an $\ell\times\ell$ square torus $T^2$, with $k=2$ logical qubits whose logical operators are the homology classes $H_1(T^2;\mathbb F_2)=\mathbb F_2^2$. The torus mapping class group $\mathrm{SL}(2,\mathbb Z)$ acts on these, generated by a $90^\circ$ rotation $S$ and a Dehn twist $T$. Both operations may be realised by the AOD movements discussed previously \Cref{sec:rotation,sec:shear}. 

As a code automorphism, the $90^{\circ}$ rotation $S$ maps $X$ and $Z$ stabilisers to stabilisers of the same type in the toric code and exchanges the horizontal and vertical logical operator, realising a logical $\textsf{SWAP}$ of the two qubits~\cite{sayginel2025fault}. Using the method of \Cref{sec:rotation}, the rotation is realised in $3\lfloor\log_2(\ell-1)\rfloor+4$ strokes.

\paragraph{Example.}
The Dehn twist $T$ mapping $(i,j)\mapsto(i+j\bmod\ell,\,j)$, realises a logical $\textsf{CNOT}$ between the two qubits. Unlike the rotation, it is not movement only and requires also a constant depth local entangling circuit~\cite{zhu2020instantaneous} to implement. Geometrically, it is realised as a shear (the primitive of \Cref{sec:shear}) of the $\ell$ rows, costing $\lfloor\log_2(\ell-1)\rfloor+1$ strokes, followed by $\ell-1$ single row strokes of length $\ell$. A total of $\lfloor\log_2(\ell-1)\rfloor+\ell$ strokes. Grouping the last $\ell -1$ single strokes requires a dynamic AOD grid shaped as a square of $(\ell-1)^2$ sites, which would cause overlaps during release with atoms from the upper triangular half sitting in the static trap, and thus cannot be done. See figure \ref{fig:dehntwist} for an illustration of the AOD movements. This construction also applies to more general hypergraph and products of cyclic codes --- the $[[90,8,10]]$ bivariate-bicycle code being one such example~\cite{tiew2025low}.

\begin{figure*}[t]
\centering
\includegraphics[width=0.95\textwidth]{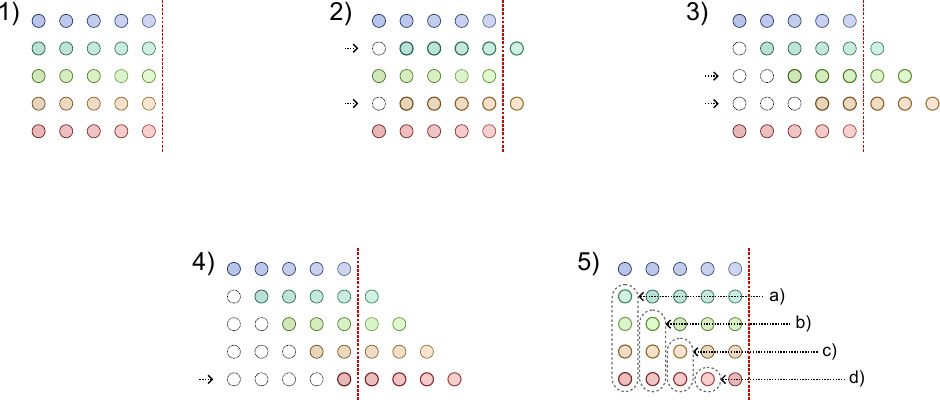}
\caption{A Dehn twist realised by binary decomposed shears.}
\label{fig:dehntwist}
\end{figure*}

\subsection{Bacon-Shor gadget for La-cross codes}
\label{sec:bsgadget}
La-cross codes are HGP codes built from a classical cyclic seed $h(x) = 1 + x + x^k$~\cite{pecorari2025high} --- a long-range generalisation of the regular surface code.
Addressable Clifford gates between La-cross logical qubits can be performed using an ancilla $d \times d$ Bacon-Shor (BS) gadget~\cite{pecorari2025addressable},  which, combined with correlated decoding techniques, enables fast logical gates in $\mathcal{O}(1)$ QEC rounds.

For instance, implementing a logical Hadamard on a given logical qubit requires: a transversal \textsf{CNOT} controlled by the BS rows and targeting $d$ non-overlapping representatives of $\bar X_L$ (the $X$ logical operator of the targeted logical qubit), a $90^\circ$ rotation of the BS patch re-aligning its rows with the La-cross columns, and a subsequent transversal \textsf{CZ} coupling them with the $d$ non-overlapping representatives of $\bar Z_L$. Measuring the BS patch in the $X$-basis completes the logical Hadamard on the targeted La-cross qubit.

Implementing the rotation step efficiently is therefore important for the speed of the whole protocol. Applying \Cref{sec:rotation} to the  rotation of the $d \times d$ BS patch yields $\Nmoves = 3\lfloor\log_2(d-1)\rfloor + 4$ strokes and $\Tjerk = O(d^{1/3})$. This allows fast logical computation on La-cross codes with neutral-atom hardware, an important ingredient for the design of future fault-tolerant quantum computers. 
 
\section{Discussion}
\label{sec:discussion}
We have given binary- and negabinary-decomposed AOD primitives for shearing, rotating, and reflecting atom arrays, each using $O(\log d)$ strokes and $O(d^{1/3})$ constant-jerk time --- a polynomial improvement over the $O(d^2)$ strokes and $O(d^{7/3})$ time of move-by-move schemes. We have demonstrated that the  primitives may be used in a number of QEC applications as data block rotation (\textit{e.g.} transversal \textsf{H}, fold-transversal \textsf{S}), automorphism gates (\textit{e.g.} $S$ and $T$ symmetries on the toric code), and ancilla rearrangement (\textit{e.g.} Bacon-Shor gadget for La-cross).

A number of directions remain for further investigation. The first is the factor $3/2$ gap in stroke count between the Paeth construction and the lower bound of Lemma \ref{lem:lowerbound}. Extensive brute-force search suggests that at $d=5$ nine strokes is optimal (realised by the array subdivision method of remark \ref{rmk:subdivision}), one above the lower bound of eight; it would thus be interesting to investigate tighter lower bounds. A second direction is to extend the primitives beyond rigid transformations to arbitrary qubit permutations. Finally, benchmarking logical fidelity against different schedules under realistic hardware noise remains for future work.


\paragraph*{Acknowledgements.}
A patent application related to the methods presented in this work has been filed by QPerfect.
We thank Shannon Whitlock and Stanimir Kondov for inspiring discussions around neutral atom hardware.
This work was supported by the French government through a CIFRE grant managed by the Association Nationale de la Recherche et de la Technologie (ANRT), convention No. 2025/1625.

\paragraph*{Author contributions.}
T.H. conceived the project and derived the main results, with additional contributions from all authors. All authors contributed to the discussions that shaped the method. T.H. and S.C. wrote the manuscript and produced the figures; A.P.O. and H.P. reviewed and revised it. A.P.O. supervised the project. Generative AI tools were used to assist with code development, text and image production, and bibliographic research. All generated content was reviewed, modified as needed, and validated by the authors. 

\bibliography{references}
\bibliographystyle{quantum}

\end{document}